\documentclass[10pt,journal,twocolumn]{IEEEtran}

\usepackage{times}
\usepackage{amssymb}
\usepackage{amsmath}
\usepackage{epsfig}
\usepackage{breqn}

\newcommand{\be}{\begin{eqnarray}}
\newcommand{\ee}{\end{eqnarray}}

\begin{document}

\renewcommand{\textfraction}{0}
\newtheorem{theorem}{Theorem}
\newtheorem{lemma}{Lemma}
\newtheorem{definition}{Definition}
\newtheorem{remark}{Remark}
\title{On Locality in Distributed Storage Systems}

        
\author{Ankit~Singh~Rawat
        and~Sriram~Vishwanath
\thanks{Ankit Singh Rawat and Sriram Vishwanath are with the Dept. of ECE, University of Texas at Austin, Austin, TX 78751 USA.\newline E-mail: \{ankitsr, sriram\}@mail.utexas.edu.}
}
\maketitle

\begin{abstract}
This paper studies the design of codes for distributed storage systems (DSS) that enable local repair in the event of node failure. This paper presents locally repairable codes based on low degree multivariate polynomials. Its code construction mechanism extends  work on Noisy Interpolating Set by Dvir et al. \cite{dvir2011}. The paper presents two classes of codes that allow node repair to be performed by contacting $2$ and $3$ surviving nodes respectively. It further shows that both classes are good in terms of their rate and minimum distance, and allow their rate to be bartered for greater flexibility in the repair process.
\end{abstract}


\IEEEpeerreviewmaketitle

\section{Introduction}

The importance of `cloud' storage has resulted in a growing body of work in both theoretical analysis and practical designs for efficient distributed storage systems. Conventionally, resilience in distributed storage is obtained by simple replication of data; however, such  replication can be highly inefficient in terms of the number of nodes required for this resilience. Thus, coding has come to play a central role in designs for resilient distributed storage systems (DSS). In particular, coding schemes for DSS must enable  efficient system repair in the event of (a small number of) node failures \cite{dimakis}. In \cite{dimakis}, Dimakis et al., the authors consider the total amount of data downloaded during single node repair, i.e., {\em repair bandwidth}, as an important metric to gauge the efficiency of any coding scheme employed in DSS and presents an lower bound on repair bandwidth. Since then, multiple codes has been proposed that achieve this lower bound.
 
In general, there are multiple, possibly apposing, metrics using which the performance of a DSS can be characterized such as security, locality, load-balancing and privacy; and the metric of interest in this paper is locality \cite{gopalan, oggier_hom, oggier_proj}. The goal of this line of research is to design coding mechanisms for DSS that enable node repair to be accomplished while requiring contact with only a small number of surviving nodes in the system. In \cite{gopalan}, Gopalan et al. establish an upper bound analogous to the singleton bound on the minimum distance of locally repairable codes and show that pyramid codes \cite{pyramid} achieve this bound. Subsequently, the work in Prakash et al. extends the bound to more general definition of locally repairable codes \cite{kumar2012}.

In this paper, our goal is to generalize \& extend the existing literature on locality in repair (and decoding) in DSS  \cite{oggier_hom, oggier_proj, pyramid}. Our coding scheme builds on schemes studied in the domain of locally decodable codes (LDC) \cite{ldc}. Specifically, our coding scheme  employs a punctured Reed-Muller (RM) code based on low-degree multivariate polynomials to store data. Our reason for choosing RM codes based on low-degree multivariate polynomials is  the inherent locality of its codewords. Indeed, it is well known in LDC literature that RM codes generated using low-degree polynomials afford locality, at the cost of being low rate. In other words, the extent of redundancy required in these codes is comparatively much higher than a code without locality properties. This has rendered an RM code based LDC unattractive, as the advantage provided by locality is superseded by the large storage space requirement of these codes.

In this paper, our approach is one of judiciously puncturing RM codes in order to obtain `good' rates for the resulting codes while still retaining the locality property for repair. A  na\"{\i}ve approach to puncturing RM codes compounds problems, as one may lose aspects of algebraic structure that make RM so desirable, including loss in structured decoding strategies as well as locality. Keeping this in mind, we turn to a methodical approach for   puncturing of  RM codes as introduced in \cite{dvir2011}. In \cite{dvir2011}, Dvir et al. develop an algorithm for puncturing RM codes based on low-degree polynomials which results in `good' codes, i.e., codes with constant rate and constant relative distance (in block length). Moreover, \cite{dvir2011} also presents an efficient decoding algorithm for these punctured codes. However, \cite{dvir2011} does not address locality properties of the resulting punctured RM code. In this paper, we show that a modified version of the punctured RM codes as studied in \cite{dvir2011} exist that are simultaneously `good' from all three perspectives -  rate (extent of storage), distance (resilience) and locality for repair of DSS.

The remainder of this paper is organized as follows. In Sec.~\ref{sec:prelim} we provide a brief introduction to generalized RM codes from polynomial evaluation perspective with their properties relevant to this paper. 
In Sec.~\ref{sec:loc}, we  define the notion of local repair and characterize the locality afforded by RM codes. In Sec.~\ref{sec:loc2} and \ref{sec:loc3}, we present two closely related coding schemes for DSS, which enable local repair based on $2$ and $3$ nodes respectively.
 
\section{Background: RM codes}
\label{sec:prelim}
A generalized RM code $\mathcal{RM}_q(u,m)$ is defined with the help of irreducible polynomials from $\mathbb{F}_{q}[x_1,\ldots, x_m]$ of degree at most $u$. Here, $\mathbb{F}_{q}[x_1,\ldots, x_m]$ denotes the ring of $m$-variate polynomials over field $\mathbb{F}_{q}$, and a polynomial in this ring is called irreducible if its degree in each variable is less than $q-1$. Throughout the paper, we assume $\mathbb{F}_q$ to be a prime field. Each irreducible polynomial of degree at most $u$ gives a $q^m$-length codeword in $\mathcal{RM}_q(u,m)$ when this polynomial is evaluated at all $q^m$ points in $\mathbb{F}_q^m$. Thus, $\mathcal{RM}_q(u,m)$ can be defined as follows:
\begin{dmath}
\label{eq:rm_def} 
\mathcal{RM}_q(u,m) = \left\{{\text{eval}(f)_{\mathbb{F}_q^m} \in \mathbb{F}_q^{q^m}:} {f \in \mathbb{F}_q[x_1,\ldots, x_m]~\text{\& deg}(f) \leq u}\right\} 
\end{dmath}

where eval$(f)_{\mathbb{F}_q^m} = (f(\alpha_1),\ldots, f(\alpha_{q^m}))$ denotes evaluations of the polynomial $f$ at all $q^m$ points in $\mathbb{F}_{q}^{m}$. The dual code of $\mathcal{RM}_q(u,m)$ is also an RM code and is generated by irreducible polynomials of degree at most $u^{\perp} = (q-1)m-u-1$, i.e.,
\begin{equation}
\label{eq:rm_dual}
\left(\mathcal{RM}_q(u,m)\right)^{\perp} = \mathcal{RM}^{\perp}_q(u,m) = \mathcal{RM}_q((q-1)m-u-1,m). \nonumber
\end{equation}

Minimum distance of $\mathcal{RM}_q(u,m)$ is given by the following  \cite{delsarte}:
\begin{equation}
\label{eq:rm_dist}
d_{min}(\mathcal{RM}_q(u,m)) = (q - \theta)q^{m - \mu - 1}
\end{equation} 
where $u = \mu(q-1) + \theta$ with $0 \leq \theta \leq q-1$. Moreover, Delsarte et al. also characterize the codewords of minimum weight for this code \cite{delsarte}. These minimum weight codewords are associated with the polynomials of following form in $\mathbb{F}_q[x_1,\ldots, x_m]$:
\begin{equation}
\label{eq:rm_min_cw}
f(\mathbf{x}) = \omega_0 \Pi_{i = 1}^{\mu} \left(1 - (\ell_i(\mathbf{x}) - \omega_i)^{q-1}\right) \Pi_{j = 1}^{\theta}\left(\ell_{\mu+1}(\mathbf{x}) - \widetilde{\omega}_j\right)
\end{equation}
where $\{\widetilde{\omega}_j\}_{j=1}^{\theta}$ are distinct elements from $\mathbb{F}_{q}$, and $\{\omega\}_{i=0}^{\mu}$ are arbitrary elements from $\mathbb{F}_q$ with $\omega_0 \neq 0$. Here, $\{\ell_i(\cdot)\}_{i=1}^{\mu + 1}$ represent $\mu + 1$ linearly independent linear forms (functions) on $\mathbb{F}_q^m$. 

\section{On Locality Properties of Codewords}
\label{sec:loc}
In this section, we illustrate the desired locality property in a codeword of a generalized RM code. First, we formally define the notion of locality of an encoded symbol:
\definition A particular symbol in a codeword has locality $r$ if it can be recovered by accessing encoded symbols from only $r$ other positions, i.e., it is uniquely defined by a set of $r$ encoded symbols in a codeword. 

For a particular encoded symbol, possessing a locality of $r$ is equivalent to having a codeword of support at most $r+1$ in the dual code such that the support of this codeword in dual code contains index of encoded symbol in interest. Since our object of study is the code $\mathcal{RM}_q(u,m)$, we focus on the codeword of minimum support in its dual code. For RM codes based on polynomials of degree at most $u \leq q-2$, we know that the following holds for $\mathcal{RM}^{\perp}_q(u,m) = \mathcal{RM}_q(m(q-1)-u-1,m)$:
\begin{equation}
u^{\perp} = (m-1)(q-1)+ (q-u-2) = \mu(q-1) + \theta. \nonumber
\end{equation}
Thus, it follows from (\ref{eq:rm_dist}) that $d_{\min}(\mathcal{RM}^{\perp}_q(u,m))$ is $u+2$. Next, we present a Lemma that establishes a necessary condition on the support of  minimum weight codewords of  $\mathcal{RM}^{\perp}_q(u,m)$.

\begin{lemma}
\label{lem:line}
For every minimum weight codeword of $\mathcal{RM}^{\perp}_q(u,m)$, $u+2$ points corresponding to its support lie on on a line in $\mathbb{F}_{q}^m$. In other words, all $u+2$ points $(\mathbf{p}_1,\ldots, \mathbf{p}_{u+2})$ are of the form $(\mathbf{p}_1, \mathbf{p}_1 + t_1\mathbf{h},\ldots, \mathbf{p}_{1}+t_{u+1}\mathbf{h})$, where $\{t_i\}_{i=1}^{u+1}$ are distinct and nonzero elements of $\mathbb{F}_q$ and $\mathbf{h}$ represents the direction of the line.
\end{lemma}
\begin{proof}
Let $\mathbf{c}$ be a minimum weight codeword of $\mathcal{RM}^{\perp}_m(u,q)$. From (\ref{eq:rm_min_cw}),  $\mathbf{c}$ is obtained as evaluations of a polynomial of the following form at all points of $\mathbb{F}_q^m$:
\begin{equation}
\label{eq:rm_min_cw1}
f(\mathbf{x}) = \omega_0 \Pi_{i = 1}^{m-1} \left(1 - (\ell_i(\mathbf{x}) - \omega_i)^{q-1}\right) \Pi_{j = 1}^{q-u-2}\left(\ell_{m}(\mathbf{x}) - \widetilde{\omega}_j\right) \nonumber
\end{equation}
where $\{\ell_i(\cdot)\}_{i=1}^{m}$ are $m$ linearly independent linear forms (functions) and  $\{\widetilde{\omega}_j\}_{j=1}^{q-u-2}$ are distinct elements in $\mathbb{F}_{q}$. 
Note that each linear form $\ell_i(\cdot)$ can be represented by a vector $L_{i} \in \mathbb{F}_q^m$ containing the coefficient of $\ell_i(\cdot)$. Let $L$ be a $m\times m$ matrix which has $L_i$ as its $i^{th}$ row. Since $\ell_i(\cdot)$ are linearly independent, $L$ is a full-rank matrix. For $u+2$ points $(\mathbf{p}_1,\ldots, \mathbf{p}_{u+2})$ corresponding to the support of $\mathbf{c}$, we have
\begin{equation}
\label{eq:lem1}
Lp_i = \Omega_i
\end{equation}
where $\left\{\Omega_i = (\omega_1, \omega_2,\ldots, \omega_{m-1}, \widetilde{\omega}_{i})^T\right\}_{i=1}^{u+2}$.
Now assume that not all $u+2$ points lie on a line, i.e., without loss of generality there exist three points in the set of $u+2$ points $\{\mathbf{p}_1, \mathbf{p}_2, \mathbf{p}_3\}$ such that
\begin{eqnarray}
\mathbf{p}_2 = \mathbf{p}_1 + t_1\mathbf{h}~\text{and}~\mathbf{p}_3 = \mathbf{p}_1 + t_2\mathbf{g} \nonumber
\end{eqnarray}
where $\mathbf{g} \neq r\mathbf{h}$ for any $r \in \mathbb{F}_q\backslash\{0\}$. From (\ref{eq:lem1}) we have
\begin{equation}
\begin{array}{l}
\label{eq:lem2}
L\mathbf{h} = \left(0,\ldots,0, \frac{\widetilde{\omega}_2 - \widetilde{\omega}_1}{t_1}\right) = \widehat{\Omega}_{\mathbf{h}}\\
L\mathbf{g} = \left(0,\ldots,0, \frac{\widetilde{\omega}_3 - \widetilde{\omega}_1}{t_2}\right) = \widehat{\Omega}_{\mathbf{g}}.
\end{array}
\end{equation}
Note that it is possible to find $(\alpha_1, \alpha_2) \neq (0,0) \in \mathbb{F}_q^2$ such that $\alpha_1\widehat{\Omega}_{\mathbf{h}} + \alpha_2\widehat{\Omega}_{\mathbf{g}} = \mathbf{0} \in \mathbb{F}_q^m$. Thus, it follows from (\ref{eq:lem2}),
\begin{equation}
L(\alpha_1\mathbf{h} + \alpha_2\mathbf{g}) = \alpha_1\widehat{\Omega}_{\mathbf{h}} + \alpha_2\widehat{\Omega}_{\mathbf{g}} = \mathbf{0}. \nonumber
\end{equation}
This, however, contradicts the full-rank nature of $L$ as $\alpha_1\mathbf{h} + \alpha_2\mathbf{g}$ is a nonzero vector and a full-rank matrix must have a trivial null space. So, all $u+2$ points corresponding to support of a minimum weight codeword $\mathbf{c}$ must lie on a line.
\end{proof}

In fact, an even stronger result holds in the sense that, given any $u+2$ points on a line in $\mathbb{F}_q^m$, a minimum weight codeword of $\mathcal{RM}_q^{\perp}(u,m)$ exists that is supported on these $u+2$ points. Next, we illustrate a procedure for determining a polynomial in $\mathbb{F}_q[x_1,\ldots, x_m]$ with degree at most $m(q-1)-u-1$ that corresponds to a codeword in $\mathcal{RM}_q^{\perp}(u,m)$ supported on a particular set of $u+2$ points on a line $(\mathbf{p}_1,\ldots, \mathbf{p}_{u+2}) = (\mathbf{p}_1, \mathbf{p}_2 + t_1\mathbf{h}\ldots, \mathbf{p}_{1}+t_{u+1}\mathbf{h})$; where $\{t_i\}_{i=1}^{u+1}$ are nonzero distinct elements of $\mathbb{F}_q$. First, pick a vector $\mathbf{v} = (0,\ldots, 0, v_m)^T \in \mathbb{F}_q^m$ which is nonzero only at the $m^{th}$ coordinate. Second, construct an invertible matrix $\mathbb{F}_q^{m\times m}$ such that 
\begin{equation}
\label{eq:mat}
M\mathbf{v} = t_1\mathbf{h}, \nonumber
\end{equation}   
i.e., the last column of $M$ is $t_1\mathbf{h}$. Finally, define a polynomial $g(\cdot)$ of degree $m(q-1)-q-1 = (m-1)(q-1) + (q-u-2)$,
\begin{equation}
\label{eq:rm_min_cw2}
g(\mathbf{x}) = \omega_0 \Pi_{i = 1}^{m-1} \left(1 - (\ell_i(\mathbf{x}) - \omega_i)^{q-1}\right) \Pi_{j = 1}^{q-u-2}\left(\ell_{m}(\mathbf{x}) - \widetilde{\omega}_j\right) \nonumber
\end{equation}
where linear functions $\ell_i(\cdot)$ are defined by the $i^{th}$ row $L_i$ of $L = M^{-1} \in \mathbb{F}_q^{m\times m}$ and $(\omega_1,\ldots,\omega_{m-1},\widehat{\omega})^T = L\mathbf{p}_1$, and $\{\widetilde{\omega}_j\}_{j=1}^{q-u-2}$ are distinct elements in $\mathbb{F}_q\backslash\{\widehat{\omega},\widehat{\omega} + \frac{t_1}{t_1}v_m, \widehat{\omega} + \frac{t_2}{t_1}v_m, \ldots, \widehat{\omega} + \frac{t_{u+1}}{t_1}v_m\}$.

It follows from Lemma~\ref{lem:line} and the construction mechanism described above that a minimum weight codeword of $\mathcal{RM}_m^{\perp}(u,q)$ can only be supported on $u+2$ points on a line; and given $u+2$ points on line, there exist multiple minimum weight codewords in $\mathcal{RM}_m^{\perp}(u,q)$ supported on these points. Note that there are multiple choices for the matrix $M$ in (\ref{eq:mat}), and each choice of $M$ results in a different codeword of $\mathcal{RM}_q^{\perp}(u,m)$ supported on $u+2$ points.

\section{Punctured RM Codes with Locality $2$}
\label{sec:loc2}
The stage is now set for a general method of designing coding schemes for DSS with locality $2$
 based on $\mathcal{RM}_q(1,m)$. Let $G$ be a generator matrix of an $[N,m,2\varepsilon N + 1]_{q}$ linear code. We construct a set $S_1 = \{\mathbf{a}_1,\ldots, \mathbf{a}_N\} \subseteq \mathbb{F}_q^{m}$, where $\{\mathbf{a}_i\}_{i=2}^N$ are $N$ rows of $G$. For each pair $(i,j) \in [N]^2$ such that $i < j$, we define a family of sets $\{B_{i,j}\}_{i<j}$ such that
\begin{equation}
\label{eq:punc_ex}
B_{i,j} = \mathbf{a}_i + t((q-1)\mathbf{a}_i+\mathbf{a}_j),~ 2 \leq t \leq 1+L, \nonumber
\end{equation}
where $L$ is a design parameter of choice. Now, we construct a set $\mathcal{S}_{\mathcal{I}}$ as follows
\begin{equation}
\label{eq:punc_set2}
\mathcal{S}_{\mathcal{I}} = S_{1} \bigcup \left(\bigcup_{(i, j) \in \mathcal{I}}B_{i,j}\right), \nonumber
\end{equation}
where $\mathcal{I} \subseteq [N]^2$ such that $(i,j) \in \mathcal{I}$ only if $i < j$. 

\subsection{Encoding Data}
Let the file to be encoded be $\mathcal{F}$. We first divide the file $\mathcal{F}$ into $K = m$ symbols $\mathbf{b} = (b_1,\ldots, b_m) \in \mathbb{F}_q^m$. Given the data vector $\mathbf{b}$, we construct a polynomial $f^{\mathbf{b}}(\cdot) \in \mathbb{F}_q[x_1,\ldots, x_m]$ of degree at most $1$ as
\begin{equation}
\label{eq:enc1}
f^{\mathbf{b}}(x_1,\ldots,x_m) = \sum_{i =1}^{m}b_ix_i.
\end{equation}
The data vector $\mathbf{b}$ is encoded to a vector $\mathbf{c}^{\mathcal{I}}_{\mathbf{b}} = \left(f^{\mathbf{b}}(\mathbf{p}_1),\ldots, f^{\mathbf{b}}(\mathbf{p}_{|\mathcal{S}_{\mathcal{I}}|})\right)$ where $\{\mathbf{p}_i\}_{i = 1}^{|\mathcal{S}_{\mathcal{I}}|}$ are distinct points of $\mathcal{S}_{\mathcal{I}}$ in any prespecified order. Each symbol in $\mathbf{c}^{\mathcal{I}}_{\mathbf{b}}$ is now stored on a different node in DSS. Let $\mathcal{C}^{\mathcal{I}}$ denote the codebook obtained using the aforementioned encoding scheme. Note that 
\begin{equation}
\label{eq:dual}
\mathcal{C}^{\mathcal{I}} \subseteq \left(\mathcal{RM}_q(1,m)\right)_{\mathcal{S}_{\mathcal{I}}},
\end{equation}
where $\left(\mathcal{RM}_q(1,m)\right)_{\mathcal{S}_{\mathcal{I}}}$ denotes the codebook obtained by puncturing a code $\mathcal{RM}_q(1,m)$ on $\mathcal{S}_{\mathcal{I}}^{C}$.

Next, we show that the code as detailed above is well defined in the sense that the dimension of the code is $K$. Let $\mathbf{y} = \left(f(\mathbf{a}_1),\ldots, f(\mathbf{a}_N)\right)^T$ denote the vector containing evaluations of $f^{\mathbf{b}}(\cdot)$ on $S_1 = \{\mathbf{a}_1,\ldots, \mathbf{a}_N\} \subseteq \mathcal{S}_{\mathcal{I}}$. It follows from (\ref{eq:enc1}) that 
\begin{equation}
\label{eq:fullrank}
\mathbf{y} = G\mathbf{b}
\end{equation}
where matrix $G$ is the generator matrix of $[N, m, 2\varepsilon N + 1]_q$ linear code used to construct $S_1$. Thus, ${\bf b}$ can be decoded from $\mathbf{y}$ using a decoding algorithm corresponding to this $[N, m, 2\varepsilon N + 1]_q$ code. This implies that the dimension of the proposed code is also $K$.

\subsection{Local Node Repair}
\label{subsec:loc2}
In this subsection, we describe a procedure for $2-$local repair of a DSS based on our code. For each $i \in [N]$, we define
$\mathcal{R}(i) = \{j \in [N] : (i,j)~\text{or}~(j,i) \in \mathcal{I} \}.$
We assume that $\mathcal{I}$ is such that 
\begin{equation}
\label{eq:assum_I}
|\mathcal{R}(i)| \geq 1~~\text{for all}~ i \in [N]
\end{equation}
i.e., for each $i \in [N]$ we add at least one set $B_{i,j}$ or $B_{j,i}$ while generating $\mathcal{S}_{\mathcal{I}}$. Without loss of generality, a node corresponding a point $\mathbf{p_1} \in \mathcal{S}_{\mathcal{I}}$ fails. It follows from (\ref{eq:assum_I}) that there exists a set of $L+2$ points $\{\mathbf{p}_1, \mathbf{p}_1 + \mathbf{h},\ldots, \mathbf{p}_1 + (L+1)\mathbf{h}\}$ in $\mathcal{S}_{\mathcal{I}}$. For example, if node corresponding to $\mathbf{a}_1$ fails and $(1,2) \in \mathcal{I}$, then $L+2$ points $\{\mathbf{a}_1, \mathbf{a}_2 = \mathbf{a}_1 + (q-1)\mathbf{a}_1 + \mathbf{a}_2,\ldots, \mathbf{a}_1 + (L+1)((q-1)\mathbf{a}_1 + \mathbf{a}_2)\}$ are in $\mathcal{S}_{\mathcal{I}}$. In this example, we have $\mathbf{p}_1 = \mathbf{a}_1$ and $\mathbf{h} = (q-1)\mathbf{a}_1 + \mathbf{a}_2$. Note that these $L+2$ points lie on a line, which has a direction $\mathbf{h}$ and passes through $\mathbf{p}_1$. For $L \geq 1$, we obtain $L+2 \geq 3$ points (including $\mathbf{p}_1$) on a line in $\mathbb{F}_q^m$. Moreover, (\ref{eq:dual}) implies that 
\begin{equation}
\label{eq:short}
(\mathcal{C}^{\mathcal{I}})^{\perp} \supseteq \left(\mathcal{RM}^{\perp}_q(1,m)\right)|^{\mathcal{S}_{\mathcal{I}}}, \nonumber
\end{equation}
where $\left(\mathcal{RM}^{\perp}_q(1,m)\right)|^{\mathcal{S}_{\mathcal{I}}}$ denotes the shortened code of $\mathcal{RM}^{\perp}_q(1,m)$ corresponding to set $\mathcal{S}_{\mathcal{I}} \subseteq \mathbb{F}_q^{q^m}$. We know from Lemma~\ref{lem:line} and the discussion following it that there exist a codeword of $\mathcal{RM}^{\perp}_q(1,m)$, which is supported on these $3$ points on a line. Moreover, this codeword is part of the shortened code $\left(\mathcal{RM}^{\perp}_q(1,m)\right)|^{\mathcal{S}_{\mathcal{I}}}$. Therefore, using this codeword in the dual code, we can recover the failed node's symbol by accessing encoded symbols corresponding to two other points on the line from two other storage nodes. This establishes $2$-locality, and therefore local repairability of our coding scheme. 

In terms of a traditional LDC understanding, the node repair process can be viewed as polynomial interpolation using at least $2$ out of the remaining $L+1$ points (excluding the point $\mathbf{p}_1$ associated with the failed node) on the line $\{\mathbf{p}_1 + \mathbf{h},\ldots, \mathbf{p}_1 + (L+1)\mathbf{h}\}$. Consider $g(t) = f^{\mathbf{b}}(\mathbf{p} + t\mathbf{h})$, a polynomial over $t$ of degree at most $1$. Given its evaluation at $2$ points, $\{\mathbf{p}_1 + t_{i_1}\mathbf{h}, \mathbf{p} + t_{i_2}\mathbf{h}\}$, we can uniquely recover $g(t)$ using any standard polynomial interpolation method. Now the desired symbol $f^{\mathbf{b}}(\mathbf{p}_1)$ can be recovered by evaluating $g(t)$ at $t = 0$.

\remark{ 
Note that, once we know the polynomial $g(t)$, we can recover encoded symbols associated with all $L+2$ points on the line defined by the pair $(\mathbf{p}_1,\mathbf{h})$. This property can be used for cooperative node repair in order to reduce repair bandwidth by determining a particular line such that it comprises of less than $L$ failures and then recovering all failures on the line simultaneously.}

\subsection{Code Parameters}
\label{subsec:code_param2}
The rate and minimum distance of the proposed coding scheme depends on three design parameters, $\mathcal{I},~L,$ and $[N, m, 2\varepsilon N + 1]_q$ linear code. In what follows, we pick an $[N, m, N-m+1]_q$ maximum distance separable (MDS) code for $[N, m, 2\varepsilon N + 1]_q$ code and analyze two cases:
\subsubsection{Case 1} 
In this case, we consider $\mathcal{I} = \{(1,2),\ldots,(i,i+1),\ldots,(N-1,N) \}$ (assuming that $N$ is even). Here, we have $|\mathcal{S}_{\mathcal{I}}| \leq N + \frac{N}{2}L$, which results in the rate of the code being greater than  $\frac{m}{N + \frac{N}{2}L} = \Theta(\frac{1}{L})$. A quick calculation that combines locality with (\ref{eq:fullrank}) shows that the proposed code is resilient against any $N-m+L$ node failures. Therefore we have,
\begin{equation}
\label{eq:dist1}
d_{\min}(\mathcal{C}^{\mathcal{I}}) \geq N-m+L+1 
\end{equation}
It also follows from (\ref{eq:fullrank}) that we can modify our code to be a systematic code by picking any set of $m$ rows of $G$ (say $\{\mathbf{a}_1,\ldots, \mathbf{a}_m\}$) to be an identity matrix, without affecting local repairability of the code. Now modify $\mathcal{I}$ to be $\{(1,2),\ldots,(m-1,m)\}$ (assuming that $m$ is even). The upper bound on $d_{\min}(\mathcal{C}^{\mathcal{I}})$ established in \cite{kumar2012} is applicable in this case. For $\delta = L + 1$ and $r = 2$ this bound results in:
\begin{eqnarray}
d_{\min}(\mathcal{C}^{\mathcal{I}}) &\leq& N + \frac{m}{2}L -m- \left(\frac{m}{2} - 1\right)L + 1\nonumber \\ &=& N-m+L+1, \nonumber
\end{eqnarray}
which, along with (\ref{eq:dist1}), proves the optimality of our codes, given that  locality for information symbols is to be ensured. Note that this code is essentially a Pyramid code as presented in \cite{pyramid,kumar2012}.
\subsubsection{Case 2}
Next we consider $\mathcal{I} = \{(i,j) \in [N]^2: i < j\}$, i.e., $|\mathcal{I}| = {N \choose 2}$. In this case, our rate becomes 
\begin{equation}
\label{eq:rate_loc2}
\text{rate}(\mathcal{C}^{\mathcal{I}}) = \frac{K}{|\mathcal{S}_{\mathcal{I}}|} \approx \frac{m}{N + {N \choose 2}L} = \Theta\left(\frac{1}{mL}\right). 
\end{equation}
The lower bound given in (\ref{eq:dist1}) holds for this case as well. For this choice of $\mathcal{I}$, many points of $\mathcal{S}_{\mathcal{I}}$ have multiple lines passing through them in $\mathcal{S}_{\mathcal{I}}$. Thus, it is more likely to be able to combine the node failures in groups along a particular line; then performing repair simultaneously for all of them by contacting just $2$ nodes. 

\remark{It is evident from previous two cases that set $\mathcal{I}$ enable us to trade-off rate of the code for flexibility in node repair and data access.}

\section{Punctured RM Codes with Locality $3$}
\label{sec:loc3}
In this section, we generalize the method of designing codes with locality $2$ from the previous section to obtain coding schemes that are $3$-local repairable. As opposed to $2$-local repairable codes, codes designed in this section are based on polynomials of degree at most $2$ in $\mathbb{F}_q[x_1,\ldots, x_m]$ and therefore related to $\mathcal{RM}_q(2,m)$. 

Let $S_1$ be the set $\{\mathbf{a}_1,\ldots, \mathbf{a}_N\} \subseteq \mathbb{F}_q^{q^m}$ as defined in Sec.~\ref{sec:loc2} with respect to an $[N, m, 2\varepsilon N + 1]_q$ linear code. In this section, we also require that the maximal hamming weight of a codeword in this $[N, m, 2\varepsilon N + 1]_q$ linear code is less than $(1-2\varepsilon)N$ and minimum distance of its dual code is at least $5$. We define another set $S_{2} = S_{1} + S_{1} \subseteq \mathbb{F}_q^{q^m}.$ The requirement on minimum distance of dual code implies that $S_1$ satisfies condition $\star_2$\footnote{
For definition of condition $\star_{2}$ and its importance for correctness of algorithm A1 (defined in Sec.~\ref{subsec:enc3}), readers may refer to \cite{dvir2011}.
}. For each pair $(i,j) \in [N]^2$ with $i < j$, we construct a family of sets $\{A_{i,j}\}_{i<j}$ similar to $\{B_{i,j}\}_{i<j}$ in Sec.~\ref{sec:loc2} such that 
\begin{equation}
\label{eq:punc_ex2}
A_{i,j} = 2\mathbf{a}_i + t((q-1)\mathbf{a}_i+ \mathbf{a}_j),~~3 \leq t \leq 2+L, \nonumber
\end{equation}
where $L$ is again a design parameter. Now, we generate a set $\mathcal{S}_{\mathcal{I}}$ as follows:
\begin{equation}
\label{eq:punc_set3}
\mathcal{S}_{\mathcal{I}} = S_{2} \bigcup \left(\bigcup_{(i,j) \in \mathcal{I}}A_{i,j}\right),
\end{equation}
where $\mathcal{I}$ is as defined in Sec.~\ref{sec:loc2}.

\subsection{Encoding data}
\label{subsec:enc3}
Given a file $\mathcal{F}$ to be encoded, we divide it into $K = {m + 2\choose m}$ symbols $\mathbf{b} = [b_1,\ldots, b_K] \in \mathbb{F}_q^K$. The data vector $\mathbf{b}$ is used to construct a polynomial $f^{\mathbf{b}}(\mathbf{x}) \in \mathbb{F}_q[x_1,\ldots, x_m]$ as follows
\begin{equation}
\label{eq:enc3}
f^{\mathbf{b}}(x_1,\ldots,x_m) = \sum_{i \in \mathcal{M}}b_i\mathbf{x}^{\alpha(i)}, 
\end{equation} 
where $\mathcal{M}$ is the index set for lexicographically arranged irreducible monomials of degree at most $2$ in $\mathbb{F}_q[x_1,\ldots, x_m]$ and $\mathbf{x}^{\alpha(i)} = x_1^{\alpha(i,1)}\ldots x_m^{\alpha(i,m}$ is $i^{th}$ monomial in $\mathcal{M}$, which is uniquely defined by its exponent vector $\alpha(i) = [\alpha(i,1),\ldots, \alpha(i,m)]$. 
In order to get codeword $\mathbf{c}_{\mathbf{b}}^{\mathcal{I}}$ corresponding to data vector $\mathbf{b}$, we evaluate the polynomial $f^{\mathbf{b}}(\mathbf{x})$ at all points in $\mathcal{S}_{\mathcal{I}}$. It follows from (\ref{eq:rm_def}) that we have $\mathcal{C}^{\mathcal{I}} = \left(\mathcal{RM}_q(2,m)\right)_{\mathcal{S}_{\mathcal{I}}}$, where $\mathcal{C}^{\mathcal{I}}$ denotes the codebook that we get from aforementioned encoding procedure.

In order to show that $\mathcal{C}^{\mathcal{I}}$ is well defined, i.e., its dimension is $K$, we can potentially utilize an approach similar to that used in Sec.~\ref{sec:loc2} and show that a sub-matrix of the generator matrix of $\mathcal{C}^{\mathcal{I}}$ is full rank. However, we follow a different approach in which we show that a polynomial interpolation algorithm recovers the data polynomial $f^{\mathbf{b}}(\mathbf{x})$, therefore the original data vector $\mathbf{b}$, from the evaluations of $f^{\mathbf{b}}(\mathbf{x})$ on $S_2 \subseteq \mathcal{S}_{\mathcal{I}}$. The interpolation algorithm is due to Dvir et al. \cite{dvir2011}, and plays an important role in establishing a lower bound on $d_{\min}(\mathcal{C}^{\mathcal{I}})$ in Sec.~\ref{subsec:code_param3}. We present an outline of the algorithm in context of recovering a polynomial of degree at most $2$. Interested readers may refer to \cite{dvir2011} for complete algorithm and its analysis.

\textbf{Interpolation Algorithm A1 \cite{dvir2011}:} For a polynomial $f(\mathbf{x})$, we define its partial derivate vector
$\Delta_f(\textbf{x}) = \left(\frac{\partial f}{\partial x_1}(\mathbf{x}),\ldots, \frac{\partial f}{\partial x_m}(\mathbf{x})\right)$
and directional derivate in the direction of $\mathbf{a} \in \mathbb{F}_q^m$
\begin{equation}
\partial_{f}(\mathbf{x},\mathbf{a}) = \sum_{i=1}^m a_i\cdot\frac{\partial f}{\partial x_i}(\mathbf{x}) \nonumber
\end{equation}
It follows from Lemma 2.1 in \cite{dvir2011} that for any $\mathbf{a}, \mathbf{b} \in \mathbb{F}_q^m$ and the polynomial of interest $f^{\mathbf{b}}(\cdot)$ of degree at most $2$,
\begin{equation}
\label{eq:lem_dvir}
f^{\mathbf{b}}(\mathbf{x}+\mathbf{a}) - f^{\mathbf{b}}(\mathbf{x} + \mathbf{b}) = \partial_{f^{\mathbf{b}}_2}(\mathbf{x},\mathbf{a}-\mathbf{b}) + E
\end{equation}
where $E$ is a constant and $\partial_{f^{\mathbf{b}}_2}(\mathbf{x})$ is a degree $1$ polynomial which represents directional derivative of $f^{\mathbf{b}}_2(\mathbf{x})$, homogeneous part of $f^{\mathbf{b}}(\mathbf{x})$ with degree $2$. Given evaluations of $f^{\mathbf{b}}(\mathbf{x})$ on $S_2$ the algorithm works as follows:

\noindent\textbf{Step 1:} Define 
\begin{equation}
\label{eq:algo1}
T(i) = S_{1} + \mathbf{a}_{i} = (\mathbf{a}_1 + \mathbf{a}_{i},\ldots, \mathbf{a}_{N} + \mathbf{a}_{i}) \nonumber
\end{equation}
and 
\begin{equation}
\label{eq:algo2}
\mathbf{c}_{\mathbf{b}}^{\mathcal{I}}(i)  = (\mathbf{c}_{\mathbf{b}}^{\mathcal{I}})_{T(i)} = (f^{\mathbf{b}}(\mathbf{a}_1 + \mathbf{a}_{i}),\ldots, f^{\mathbf{b}}(\mathbf{a}_{N} + \mathbf{a}_{i})) \nonumber
\end{equation}
Note that it follows from (\ref{eq:lem_dvir}) that
$\mathbf{c}_{\mathbf{b}}^{\mathcal{I}}(i) - \mathbf{c}_{\mathbf{b}}^{\mathcal{I}}(j) =(f^{\mathbf{b}}(\mathbf{a}_1 + \mathbf{a}_{i}) -f^{\mathbf{b}}(\mathbf{a}_1 + \mathbf{a}_{j}),\ldots, f^{\mathbf{b}}(\mathbf{a}_{N} + \mathbf{a}_{i})-f^{\mathbf{b}}(\mathbf{a}_{N} + \mathbf{a}_{j})) $
represents evaluations of $g_{i,j}(\mathbf{x}) = f^{\mathbf{b}}(\mathbf{x} + \mathbf{a}_i) - f^{\mathbf{b}}(\mathbf{x} + \mathbf{a}_j)$, a polynomial of degree at most $1$, on $S_1$. Thus, using decoding algorithm for $[N,m,2\varepsilon N + 1]_q$ code, we can recover $g_{i,j}(\mathbf{x})$ as proved in part $2$ of Lemma 2.3 in \cite{dvir2011}. Removing the constant term within $g_{i,j}(\mathbf{x})$ results in $\partial_{f^{\mathbf{b}}_2}(\mathbf{x},\mathbf{a}_i-\mathbf{a}_j)$. 

\noindent\textbf{Step 2:} The algorithm takes all ${N \choose 2}$ $\{\partial_{f^{\mathbf{b}}_2}(\mathbf{x},\mathbf{a}_i-\mathbf{a}_j)\}_{i<j}$ and recover $\Delta_{f^{\mathbf{b}}_2}(\mathbf{x})$. As a homogeneous polynomial can be recovered from its partial derivative vector \cite{dvir2011}, we get $f^{\mathbf{b}}_2(\mathbf{x})$ from $\Delta_{f^{\mathbf{b}}_2}(\mathbf{x})$.

\noindent\textbf{Step 3:}  Subtract the contribution of evaluations of $f^{\mathbf{b}}_2(\mathbf{x})$ from codeword, and recover degree $1$ polynomial $\xi^{\mathbf{b}}(\mathbf{x})= f^{\mathbf{b}}(\mathbf{x}) - f^{\mathbf{b}}_{2}(\mathbf{x})$ using the decoding algorithm for $[N, m, 2\varepsilon N + 1]_q$ code that generates $S_1$. Output coefficients of $f^{\mathbf{b}}(\mathbf{x}) = \xi^{\mathbf{b}}(\mathbf{x}) + f^{\mathbf{b}}_{2}(\mathbf{x})$ as the original data vector.

\subsection{Local Node Repair}
In this subsection, we explain that under the assumption similar to (\ref{eq:assum_I}) on $\mathcal{I}$, all symbols in $\mathcal{C}^{\mathcal{I}}$ have locality at most $3$, i.e., each failed storage node can be recovered by contacting $3$ storage nodes. Since $\mathcal{C}^{\mathcal{I}} = (\mathcal{RM}_q(2,m))_{\mathcal{S}_{\mathcal{I}}}$, we have $\left(\mathcal{C}^{\mathcal{I}}\right)^{\perp} = (\mathcal{RM}^{\perp}_q(2,m))|^{\mathcal{S}_{\mathcal{I}}}$. 

Following the reasoning used in Sec.~\ref{subsec:loc2} with assumption that $\mathcal{I}$ adds at least one local parity for each point for each point $\mathbf{p} \in \mathcal{S}_{\mathcal{I}}$, we can find a set of $L+3$ points $\{\mathbf{p}, \mathbf{p}+t_1\mathbf{h},\ldots, \mathbf{p}+t_{L+2}\mathbf{h}\}$ that lie on a line in $\mathbb{F}_q^{q^m}$. We know from Lemma~\ref{lem:line} and the discussion that follows the Lemma that for each set of $4$ points on this line there exist a codeword of weight $4$ in $(\mathcal{C^{\mathcal{I}}})^{\perp} = (\mathcal{RM}^{\perp}_q(2,m))|^{\mathcal{S}_{\mathcal{I}}}$ supported on the $4$ points under consideration. Therefore, for each encoded symbol $f^{\mathbf{b}}(\mathbf{p}), \mathbf{p} \in \mathcal{S}_{\mathcal{I}}$, we can locally recover $f^{\mathbf{b}}(\mathbf{p})$ by contacting a set of $3$ nodes storing symbols associated with $3$ points in the aforementioned set of $L+3$ points on a line defined by the pair $(\mathbf{p},\mathbf{h})$. In fact, $f^{\mathbf{b}}(\mathbf{p})$ can be recovered without knowing the dual codeword of weight $4$ supported on $4$ points (including $\mathbf{p}$) by applying polynomial interpolation based local decoding algorithm described in Sec.~\ref{subsec:loc2}. This establishes that $\mathcal{C}^{\mathcal{I}}$ is $3$-local repairable.

\subsection{Code Parameters}
\label{subsec:code_param3}
In this subsection, we study the rate and minimum distance of $\mathcal{C}^{\mathcal{I}}$. 
For $m$ large enough and under the assumption that $[N, m, 2\varepsilon N+1]_q$ code satisfies requirements specified in the beginning of Sec.~\ref{sec:loc3}, it follows from Theorem 1.5 in \cite{dvir2011} that algorithm A1 recovers the data polynomial $f^{\mathbf{b}}(\mathbf{x})$ from its evaluations on $S_2$ even when $\frac{\varepsilon^2}{18}|S_2|$ evaluations are incorrect. Therefore we have
\begin{equation}
\label{eq:lw_dist3}
d_{\min}(\mathcal{C}^{\mathcal{I}}) \geq \frac{\varepsilon^2}{9}|S_2| +1 = \Theta(m^2),
\end{equation}
Similar to Sec.~\ref{subsec:code_param2}, we present rate of $\mathcal{C}^{\mathcal{I}}$ for two choices for $\mathcal{I}$:
\subsubsection{Case 1}
Here we take $\mathcal{I} = \{(i,j) \in [N]^2: i<j\}$. This ensures that each symbol has at least one set of $L+3$ points to allow its local repair. Note that in this case, for some encoded symbols there are multiple line passing through these symbols in $\mathcal{S}_{\mathcal{I}}$. For example, each symbol corresponding to a point in $\{2\mathbf{a_i}\}_{i=1}^{N}$ can be repaired along $N-1$ lines. In this case, we have that $|\mathcal{S}_{\mathcal{I}}| \leq (L+1)N^2$. Therefore,
\begin{equation}
\label{eq:rate_loc3}
\text{rate}(\mathcal{C}^{\mathcal{I}}) \approx \frac{(1-2\varepsilon)^2}{2L} = \Theta \left(\frac{1}{L}\right).
\end{equation}
which can be considered as a good rate for locally repairable codes, when $L$ is small.

\subsubsection{Case 2}
Present definition of $\mathcal{I}$ and $\mathcal{S}_{\mathcal{I}}$ utilize the fact that $\{2\mathbf{a}_i,\mathbf{a}_i+\mathbf{a}_j,2\mathbf{a}_j\}$ lie on a line. We may modify the set $\mathcal{I}$ to be a subset of $[|S_2|]^2$ such that $(i,j) \in \mathcal{I}$ only if $i < j$. The set $\mathcal{S}_{\mathcal{I}}$ also needs to be modified accordingly. Let $\{\mathbf{p}_i\}_{i=1}^{|S_2|}$ be points of $S_2$ in a prespecified order. Take $\mathcal{I} = \{(1,2),\ldots,(i,i+1),\ldots, (|S_2|-1,|S_2|)\}$ (assuming $|S_2|$ is even). Take a family of sets 
\begin{equation}
\label{eq:punc_case}
C_{i,j} = \mathbf{p}_i + t((q-1)\mathbf{p}_i+ \mathbf{p}_j),~~2 \leq t \leq 2+L, \nonumber
\end{equation}
Now, we generate a set $\mathcal{S}_{\mathcal{I}}$ as follows:
\begin{equation}
\label{eq:punc_set33}
\mathcal{S}_{\mathcal{I}} = S_{2} \bigcup \left(\bigcup_{(i,j) \in \mathcal{I}}C_{i,j}\right),
\end{equation}
With these choices of $\mathcal{I}$ and $\mathcal{S}_{\mathcal{I}}$, each node has one set of $L+3$ points to exploit for local repairablity, and it translates into rate of $\mathcal{C}^{\mathcal{I}}$ being greater than $\frac{K}{N^2 + \frac{N^2}{2}(L+1)}$. Note that this rate is at least that in previous case. 

\section*{Acknowledgment}
The authors would like to thank Natalia Silberstein for valuable discussions.

\bibliographystyle{IEEEtran}
\bibliography{locality_final}

\end{document}